%% file: main.tex
\documentclass[runningheads,orivec]{llncs}

\usepackage[T1]{fontenc}
\usepackage{graphicx}
\graphicspath{{figs/}}
\usepackage{amsmath,amssymb}
\usepackage{booktabs}
\usepackage[table]{xcolor}
\PassOptionsToPackage{hyphens}{url}
\usepackage[hidelinks]{hyperref}
\usepackage{bm}
\usepackage{orcidlink}
\newcommand{\best}[1]{\textbf{\boldmath #1}}
\newcommand{\snd}[1]{\underline{#1}}
\definecolor{heroblue}{HTML}{4C72B0}    % figure hero colour (our method)
\definecolor{ourstint}{HTML}{E8EEF6}    % light tint for the "ours" table row
\definecolor{cmarkgreen}{HTML}{55A868}  % figure green for a "yes" tick
\newcommand{\ours}{\textcolor{heroblue}{\textbf{(Ours)}}}
\newcommand{\cmark}{\textcolor{cmarkgreen}{\checkmark}}

\begin{document}

\title{Manipulation-Resistant Oracles via Continuous Peer Prediction}
\titlerunning{Manipulation-Resistant Oracles via Continuous Peer Prediction}

% == authors ==
\author{
Junhao Song$^\dagger$\inst{1}\orcidlink{0009-0009-5178-0258} \and
Tanis Rolandi\inst{2}\orcidlink{0009-0006-0055-4267} \and
Adam Novocky\inst{1}\orcidlink{0009-0007-5015-244X} \and
\\
Enrico Biffis$^\dagger$\inst{3,4}\orcidlink{0000-0001-5984-0942} \and
William Knottenbelt$^\dagger$\inst{1}\orcidlink{0000-0002-8490-1011}
}

\institute{
Department of Computing, Imperial College London, London, UK \and
European Research and Development Centre, Hitachi Europe Ltd., London, UK \and
Department of Finance, Business School, Imperial College London, London, UK \and
Centre for Climate Finance and Investment, Imperial College London, London, UK\\
\email{\{junhao.song23,adam.novocky26,e.biffis,w.knottenbelt\}@imperial.ac.uk},
\email{tanis.rolandi@hitachi-eu.com}\\
$\dagger$ Corresponding authors}

\authorrunning{J. Song et al.}

\maketitle

\begin{abstract}
Blockchain oracle networks, including decentralised physical infrastructure networks (DePINs), report physical measurements from a source that can manipulate the underlying process. A carbon project operator, for example, can dry a site before measurement and be paid for soil organic carbon (SOC) that was never stored. Peer prediction covers discrete signals under a strategic source, or continuous reports under an honest source, but not both at once. We close that gap for Gaussian reports. Strict propriety of the continuous ranked probability score (CRPS) settles the reporting stage alone. Our result concerns the source stage. A report-only settlement rule whose honest-source payoff is non-decreasing in the true value identifies and deters source manipulation if and only if the sensors respond to that manipulation by different amounts. When they respond identically, manipulation is observationally equivalent to a genuine rise in SOC, and CRPS translation equivariance leaves every peer score unchanged, so no such rule deters the source. Necessity follows from invariance rather than from the signal enumeration available in the discrete case. The source-stage estimator is ordinary weighted least squares, and the contribution is the characterisation around it. With eight calibrated SOC models on $37{,}280$ real European soil spectra, shift-naive aggregators over-credit by up to $57\%$ under a calibrated moisture perturbation while the joint estimate stays near-unbiased. The guarantee holds on the equilibrium path for non-colluding reporters and inherits the models' own calibration error. The project page is at \url{https://rqcker.github.io/crps-wls}.

\end{abstract}

\input{1_intro}
\input{2_related}
\input{3_model}
\input{4_method}
\input{5_experiments}
\input{6_conclusion}

\section*{Acknowledgements}
This project was funded by the Hitachi-Imperial Centre for Decarbonisation and Natural
Climate Solutions, a collaboration between Hitachi Ltd, Hitachi Europe and Imperial
College London, which fully funds Junhao Song's PhD studentship. We thank the Imperial
College Centre for Cryptocurrency Research and Engineering and the Brevan Howard Centre
for Financial Analysis for their support. Adam Novocky contributed to this work as a
PhD intern at Imperial College London.

\bibliographystyle{splncs04}
\bibliography{refs}

% Appendices are not permitted in the MARBLE 2026 proceedings version.
% The extended version restores them with:
\appendix \input{7_appendix}

\end{document}

%% file: 1_intro.tex
\section{Introduction}
\label{sec:intro}

Carbon credit markets pay landowners for storing carbon in soil, and the payment
depends on a \textbf{\underline{S}}oil \textbf{\underline{O}}rganic \textbf{\underline{C}}arbon (\textbf{SOC}) measurement that the buyer cannot
independently check \cite{oldfield2022crediting}.
Soil moisture changes the water absorption bands of a spectrum and biases the
predicted SOC \cite{minasny2011removing}.
Seven of the eight models calibrated in Section \ref{sec:experiments} read higher on
drier spectra, so an operator who dries the site before measurement is paid for
carbon that is not there.
The party that controls the measurement site is the source of the signal, with both
the means and the motive to manipulate it (Figure \ref{fig:teaser}).

\begin{figure}[t]
  \centering
  \includegraphics[width=\linewidth]{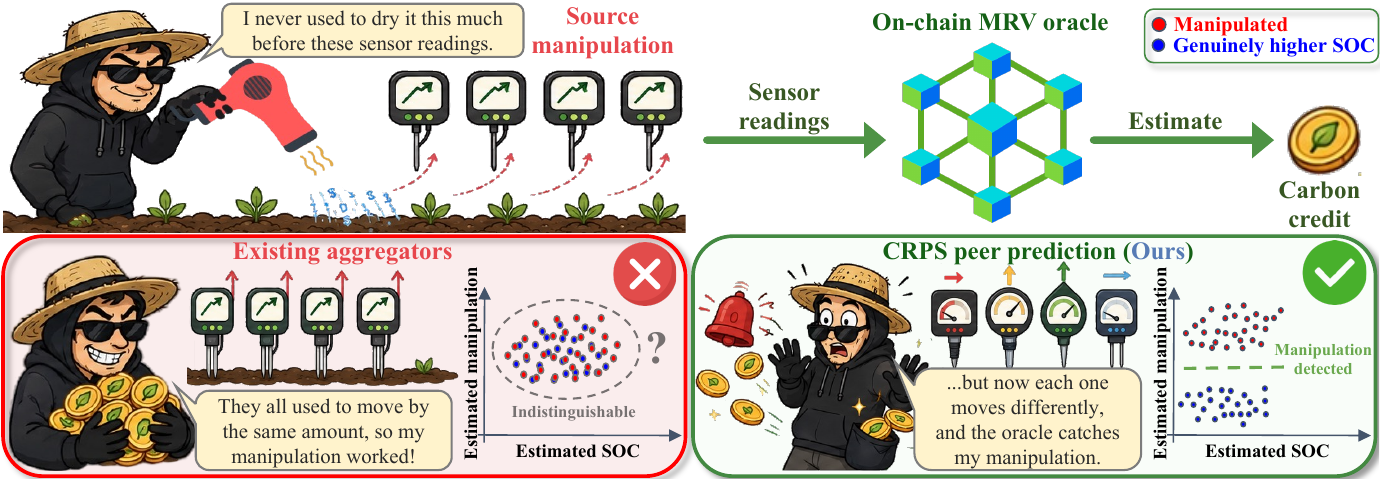}
  \caption{\textbf{Sensor diversity separates manipulation from real carbon.}
  Drying the site shifts every sensor's reading. A uniform shift is indistinguishable from
  genuinely higher carbon (\textcolor[HTML]{FF0000}{\textbf{left}}), but when sensor
  types respond by different amounts the pattern of readings reveals the manipulation
  (\textcolor[HTML]{438941}{\textbf{right}}).}
  \label{fig:teaser}
\end{figure}

This failure is an instance of a broader one. Oracle manipulation has drained hundreds
of millions from decentralised finance, including over \$110 million from Mango Markets
in 2022 \cite{cftc2023eisenberg}. Those attacks corrupt a reported value. Source
manipulation works one step earlier: it changes the physical quantity before
measurement, so the sensor models report honestly and the oracle relays a true reading
of an engineered reality.

Blockchain settles the carbon payment trustlessly but cannot see the soil.
A \textbf{\underline{M}}easurement, \textbf{\underline{R}}eporting, and \textbf{\underline{V}}erification (\textbf{MRV}) oracle brings the off-chain reading
on-chain, and \textbf{\underline{De}}centralised \textbf{\underline{P}}hysical \textbf{\underline{I}}nfrastructure \textbf{\underline{N}}etworks (\textbf{DePINs}) such as Open
Forest Protocol (OFP) \cite{ofp2024}, WeatherXM \cite{weatherxm2023}, and Boomitra
\cite{boomitra2025} already run them. Each aggregates readings by majority vote,
self-consistency, or a centralised model, and none scores the source for honesty. A
source that shifts every sensor by the same amount defeats them all, since a uniform
shift is indistinguishable from genuinely higher carbon.

Four roles run through the paper. The \emph{source} owns the site and chooses the
physical intervention. A \emph{sensor model} turns the resulting measurement into a
calibrated Gaussian report about the SOC value. A \emph{reporter} submits that report
and collects a peer payment. The \emph{settlement rule} reads the submitted reports,
issues credits, and slashes the source collateral. The source and the reporters are
the strategic parties; a sensor model is the calibrated pipeline a reporter runs.

Peer prediction scores an agent by how well its report predicts a peer's, so truthful
reporting is an equilibrium with no ground truth \cite{miller2005eliciting}.
Two lines of work approach our setting but do not meet in it.
Milionis et al. \cite{milionis2025incentive} add a manipulating source, but only for
discrete signals.
Richardson and Faltings \cite{richardson2024peer} handle the continuous distributional
reports a calibrated model produces, but assume an honest source.
SOC measurement needs both at once.
Milionis et al. note that their discrete result can be generalised to continuous
signals, and we carry out that extension for the Gaussian case.

Strict propriety of the \textbf{\underline{C}}ontinuous \textbf{\underline{R}}anked \textbf{\underline{P}}robability \textbf{\underline{S}}core (\textbf{CRPS})
\cite{gneiting2007strictly} settles the reporting stage. A reporter scored on how well
its distribution predicts a peer's report can do no better than state that distribution
truthfully. Propriety says nothing about whether the source action is visible in the
report law. That question is the source stage. Manipulation enters there as a nuisance
direction, and the joint estimate of the SOC value and the manipulation is ordinary
weighted least squares. The contribution is the characterisation around that estimate.
Identification and deterrence by report-only settlement rules hold exactly when the
sensitivity vector is transverse to the all-ones direction. The necessity argument
cannot reuse the discrete proof, which enumerates a finite signal alphabet. It rests
instead on an invariance. CRPS is translation equivariant, so a common shift of every
report leaves all peer scores unchanged. Uniform sensitivities then make manipulation
observationally equivalent to a genuine rise in SOC. We call the transversality
requirement the \textbf{SOC Sensor Diversity Condition} and show that it is the exact
dividing line. Diversity is a requirement rather than a modelling convenience, and it
is testable, since different model families respond differently to the same physical
change. The condition is stated here for one manipulation direction. An adversary free
to choose the intervention needs diversity against every feasible direction, the
multi-mode condition of Section \ref{sec:method}.

\medskip\noindent Our contributions are as follows:
\begin{itemize}
  \item We give a peer prediction mechanism for continuous Gaussian reports under an
  adversarial source, the case left open between discrete strategic-source and
  continuous honest-source peer prediction.
  \item We give an if-and-only-if characterisation (Theorem \ref{thm:main}): source
  manipulation is identifiable, and deterrable by report-only rules whose
  honest-source payoff is non-decreasing in the true value, exactly when the sensors
  are diverse. When diversity fails, the translation equivariance of CRPS makes
  manipulation indistinguishable from genuinely higher carbon, and no such rule
  deters it.
  \item On $37{,}280$ real European soil spectra with eight calibrated SOC models,
  shift-naive aggregators over-credit by up to $57\%$ under the calibrated moisture
  perturbation, while the joint estimate stays near-unbiased and about seven times
  less sensitive. We report the price too: a higher honest-case error floor, modest
  detection power, and no paired wet-and-dry validation.
\end{itemize}

%% file: 2_related.tex
\section{Related Work}
\label{sec:related}

\noindent\textbf{Peer prediction.}
Peer prediction mechanisms elicit truthful reports from agents who observe
correlated private signals without access to a ground truth.
Miller et al. \cite{miller2005eliciting} introduced the peer prediction method, which
scores an agent by how well its report predicts a peer's signal under a common prior.
It makes truthful reporting an equilibrium, though not the only one, and later work
strengthens the guarantee.
Schoenebeck and Yu \cite{schoenebeck2023two} propose the Source and Target
\textbf{\underline{D}}ifferential \textbf{\underline{P}}eer \textbf{\underline{P}}rediction (\textbf{DPP}) mechanisms, strongly truthful
for three agents on a single question, with prior-free guarantees but finite signal
reports and an honest source.
Zhao et al. \cite{zhao2024takes} apply peer prediction to the blockchain
verifier's dilemma for binary signals, and Wang et al. \cite{wang2021forecast}
use it for forecast aggregation. Neither models a manipulating source.
Zhang et al. \cite{zhang2025stochastically} propose stochastic dominance
truthfulness, a stronger criterion, but establish it for binary signals only.
Multi-task peer prediction attains informed truthfulness that resists uninformed
and colluding strategies \cite{dasgupta2013crowdsourced,shnayder2016informed},
but a single physical site offers only one task.
None of these works address continuous distributional reports or source
manipulation of a physical quantity.

\noindent\textbf{Continuous and adversarial peer prediction.}
Two lines come closest to our setting but stop short of it.
Richardson and Faltings \cite{richardson2024peer} introduce peer neighbourhood
mechanisms, which extend peer prediction to arbitrary signal distributions. They
achieve strict Bayes--Nash equilibria for continuous and distributional reports,
including the Gaussian reports that calibrated SOC models produce, but they assume
an honest source.
Frongillo et al. \cite{frongillo2025binary} study agents who hold more private
signals than they report, yet restrict reported outputs to binary values.
Milionis et al. \cite{milionis2025incentive} add a manipulating source and extend the
Source-DPP framework to DePIN settings. Their source identifiability condition is
necessary and sufficient for a strictly truthful mechanism, but the signal model is
discrete and the proof enumerates a finite alphabet. They note that the result can be
generalised to continuous signals, and we carry out that extension for Gaussian
signals, replacing the enumeration by the translation equivariance of CRPS
\cite{gneiting2007strictly}. That score sits within the broader theory of property
elicitation \cite{lambert2008eliciting}, and Neyman and Roughgarden
\cite{neyman2023proper} show it yields minimax-optimal forecast aggregation under
truthful reporting. Strict propriety governs the reporting stage in all of this
work. What is open is whether the source action is visible in the report law at all.

\noindent\textbf{Robust and Byzantine-resilient aggregation.}
A separate line defends aggregation against corrupted inputs without modelling
incentives. Robust statistics replace the mean with high-breakdown estimators
\cite{huber2025robust}, and Byzantine-resilient rules such as Krum
\cite{blanchard2017machine} and Bulyan \cite{guerraoui2018hidden} filter outlying
updates. These assume a corrupted minority and a clean majority, so they suppress a few
deviating sensors. Source manipulation breaks that premise: a common shift moves every
sensor at once and no clean majority remains.

\begin{table}[t]
\centering
\caption{\textbf{Related work compared on the five properties used here.} Signal
type, report type, source-manipulation modelling, scoring
rule and deployment context. CRPS denotes the continuous ranked probability score.}
\label{tab:related}
\scriptsize
\setlength{\tabcolsep}{5pt}\renewcommand{\arraystretch}{1.15}
\begin{tabular}{@{}lccccc@{}}
\toprule
Method & Continuous & Distributional & Source & CRPS & Blockchain \\
\midrule
Miller et al. \cite{miller2005eliciting}          & --  & --  & --  & --  & --  \\
Schoenebeck \& Yu \cite{schoenebeck2023two}  & --  & --  & --  & -- & -- \\
Richardson \& Faltings \cite{richardson2024peer} & \cmark & \cmark & -- & -- & -- \\
Milionis et al. \cite{milionis2025incentive} & --  & --  & \cmark & -- & \cmark \\
Neyman \& Roughgarden \cite{neyman2023proper}& \cmark & \cmark & -- & \cmark & -- \\
Frongillo et al. \cite{frongillo2025binary}    & \cmark & --  & --  & --  & -- \\
Deployed DePIN oracles \cite{ofp2024,weatherxm2023,boomitra2025} & --  & --  & --  & --  & \cmark \\
\midrule
\rowcolor{ourstint}
\textbf{CRPS-WLS Oracle} \ours                  & \cmark & \cmark & \cmark & \cmark & \cmark \\
\bottomrule
\end{tabular}
\end{table}

\noindent\textbf{Soil spectroscopy and moisture compensation.}
Soil visible-to-near-infrared spectroscopy is an established low-cost proxy for SOC
whose accuracy depends on calibration design and acquisition conditions
\cite{stenberg2010visible}. Moisture is a documented nuisance there, and a line of work
removes it before prediction: external parameter orthogonalisation
\cite{minasny2011removing,wijewardane2016moisture} and direct standardisation
\cite{roudier2017evaluation} project the water response out of the spectrum, and a
trusted moisture reading can enter the calibration as a covariate. These methods act
upstream of the oracle and need reference data, paired wet and dry spectra, or a
moisture sensor the source does not control. They reduce a known measurement nuisance.
They do not score a party free to choose the measurement condition. The two layers are
complementary. Better moisture compensation feeds better calibrated reports into
the mechanism studied here.

\noindent\textbf{Blockchain oracles for carbon MRV.}
As attention to the circular economy and the low-carbon transition grows
\cite{song2024exploring}, on-chain carbon verification needs an oracle that is
both accurate and manipulation-resistant. Open Forest Protocol \cite{ofp2024} votes on a reported
value through multi-validator output agreement, WeatherXM \cite{weatherxm2023}
scores sensors by self-consistency (its peer-comparison check is planned but not
deployed), and Boomitra \cite{boomitra2025} issues credits from a centralised
model without incentive guarantees. Financial oracle designs harden price feeds with gas-efficient or volume-weighted
medians \cite{bai2024ormer,cfbenchmarks2023}, but neither elicits honest reports from a
strategic source, and industry guidance treats oracle security as an economic
cost-of-attack condition rather than an absolute guarantee \cite{chainscore2025myth}.
Cong et al. \cite{cong2025primer} frame an oracle trilemma among scalability,
decentralisation, and truthfulness, and the present mechanism targets its truthfulness
corner under the SOC Sensor Diversity Condition. Table \ref{tab:related} positions it
against the closest prior results.
The settlement layer here is a closed-form rule, not a systems contribution. It does
not establish that a sample came from the claimed parcel, enforce the measurement
protocol, or govern how the sensitivities and noise levels are calibrated. We
report no contract implementation or gas cost.

%% file: 3_model.tex
\section{System Model}
\label{sec:model}

\subsection{Sensor Modelling}
Each of $N \ge 2$ sensor models estimates the SOC of a common site, a latent
scalar $\theta \in \mathbb{R}$ (in g/kg).
We use the flat-prior Gaussian location model as the baseline because it gives the
cleanest reporting rule, $m_i=Y_i$. Section \ref{sec:invariance} gives the
proper-prior variant, under which identification and deterrence remain exact.

A self-interested source, the party that owns the site and earns carbon credits,
chooses a manipulation action $s \in \mathcal{A}_S = \mathbb{R}_{\ge 0}$ before
measurement, a physical intervention such as drying the soil, with $s = 0$ the
honest source.
The scalar action $s$ represents a fixed, calibrated manipulation direction. If the
source can choose among several physical interventions, the model gains one
sensitivity column per intervention and the relevant condition becomes the
multi-mode rank condition stated in Section \ref{sec:method}. The one-dimensional
model below should therefore be read as local to the calibrated wetting and drying
direction used in the experiments.
Because $s$ acts on the site, it reaches every sensor model through the same physical
event, which is what makes it a single column of the design below. Tampering with
individual reports is a different channel: a free per-report shift $s_i$ adds one
unknown per report and is not identifiable from one round of $N$ reports without
further structure, so it lies outside the guarantee (Section \ref{sec:security}).

Each sensor model is a calibrated machine learning predictor, and a reporter
submits its output.
Sensor model $i$ observes:
\begin{equation}
  Y_i \;=\; \theta + \alpha_i\, s + \varepsilon_i,
  \qquad \varepsilon_i \sim \mathcal{N}(0,\sigma_i^2),
  \label{eq:obs}
\end{equation}
where $\sigma_i > 0$ is the model's calibrated noise level.
For the Gaussian likelihood we take $\boldsymbol\varepsilon\sim\mathcal{N}(0,\Sigma)$
with $\Sigma_{ii}=\sigma_i^2$. The errors need not be independent, since the models
read a common spectrum, but the displayed peer score uses the diagonal case for
clarity. A general covariance is algebraically identical for identification, which
depends only on the mean map $(\theta,s)\mapsto\theta\mathbf{1}+s\boldsymbol{\alpha}$,
and affects only precision, detection power, and the stake.
The coefficient $\alpha_i > 0$ is the \textbf{manipulation sensitivity} of sensor
model $i$, and different architectures aggregate spectral information differently, so
the $\alpha_i$ need not coincide.
We collect the parameters into vectors
$\boldsymbol{\alpha} = (\alpha_1,\dots,\alpha_N)^\top$ and
$\boldsymbol{\sigma} = (\sigma_1,\dots,\sigma_N)^\top$, both common knowledge.
The values of $\boldsymbol{\alpha}$ and $\boldsymbol{\sigma}$ are calibration inputs
to the settlement rule, not strategic reports. They are public by design, and the
estimator and the penalty are built for those public values.
Misspecification of $\boldsymbol{\alpha}$ is treated separately in
Section \ref{sec:method}.

Each reporter states a Gaussian belief about $\theta$, the calibrated output of the
models we consider \cite{richardson2024peer}: a pair $R_i = (m_i, v_i)$ encoding
$\mathcal{N}(m_i, v_i)$.
In the diffuse-prior limit reporter $i$'s posterior for
$\theta$ given $Y_i$ is $\mathcal{N}(Y_i, \sigma_i^2)$, so truthful reporting is
$R_i = (Y_i, \sigma_i^2)$, and then:
\begin{equation}
  \mathbb{E}[m_i \mid \theta, s] \;=\; \theta + \alpha_i\, s .
  \label{eq:truthmean}
\end{equation}
Under a proper common Gaussian prior the truthful report is the posterior belief
instead, and Section \ref{sec:invariance} gives the fixed de-shrinking map that
recovers an unbiased signal proxy before aggregation. We write $R=(R_1,\dots,R_N)$ for
the report profile and $L_{\theta,s}$ for the joint law of $R$ at true value $\theta$
and source action $s$ under truthful reporting.

\subsection{CRPS Peer Scoring}

The mechanism scores reports with the CRPS, the strictly proper scoring rule for
full predictive distributions \cite{gneiting2007strictly}.
For a predictive distribution with a cumulative distribution function $F$ and an
observation $y$,
\begin{equation}
  \mathrm{CRPS}(F, y) \;=\; \int_{-\infty}^{\infty}\!\big(F(z) - \mathbf{1}\{z \ge y\}\big)^2\, dz .
  \label{eq:crps}
\end{equation}
For a Gaussian forecast $\mathcal{N}(m,\sigma^2)$ with $z = (y-m)/\sigma$ this has
the closed form:
\begin{equation}
  \mathrm{CRPS}\big(\mathcal{N}(m,\sigma^2), y\big)
    = \sigma\Big[\,z\big(2\Phi(z)-1\big) + 2\phi(z) - \tfrac{1}{\sqrt{\pi}}\,\Big],
  \label{eq:crps_gauss}
\end{equation}
where $\phi$ and $\Phi$ are the standard normal density and distribution
functions.
Reporter $i$ is scored against a peer $j \ne i$ by how well it predicts the peer's
observation, in the on-path, flat-prior, diagonal-covariance form used in the baseline
theorem. At the candidate equilibrium the source plays $s=0$, so a reporter whose
belief about $\theta$ is $\mathcal{N}(m_i, v_i)$ predicts $Y_j$ as
$\mathcal{N}(m_i, v_i + \sigma_j^2)$, since $Y_j$ then adds independent noise of the
known variance $\sigma_j^2$ to $\theta$ in \eqref{eq:obs}. The mechanism therefore
inflates the reported variance by that known $\sigma_j^2$ and scores the result
against the peer's reported mean $m_j$:
\begin{equation}
  S_{ij}(R)
    \;=\; \mathrm{CRPS}\big(\mathcal{N}(m_i,\, v_i + \sigma_j^2),\; m_j\big)
    \;=\; \kappa\Big[\,d\big(2\Phi(d)-1\big) + 2\phi(d) - \tfrac{1}{\sqrt{\pi}}\,\Big] ,
  \label{eq:peercrps}
\end{equation}
where the right-hand closed form specialises \eqref{eq:crps_gauss} to the truthful
variance $v_i = \sigma_i^2$, with $\kappa = \sqrt{\sigma_i^2+\sigma_j^2}$ and
$d = (m_i-m_j)/\kappa$.
The payment to reporter $i$ is $\,-\tfrac{1}{N-1}\sum_{j\ne i} S_{ij}(R)$, so a lower
peer CRPS yields a higher payment. The solution concept is the on-path Bayes--Nash
equilibrium made precise in Section \ref{sec:method}.

\subsection{Sensor Diversity}

The source is the third player, in the spirit of source differential peer
prediction \cite{schoenebeck2023two,milionis2025incentive}. After scoring the
reporters, the mechanism aggregates the reports into an SOC estimate that sets the
credits and penalises the source by an amount that depends on the reports. For the
source to face a strict incentive to stay honest, the mechanism must tell
manipulation apart from a genuinely higher SOC value, which is possible only if
distinct pairs $(\theta,s)$ induce distinct report laws.
\begin{definition}[SOC Sensor Diversity]\label{def:diversity}
The ensemble satisfies the \textbf{SOC Sensor Diversity Condition} if the report law
identifies the latent value and the source action jointly, that is,
$L_{\theta,s}=L_{\theta',s'}$ implies $(\theta,s)=(\theta',s')$.
\end{definition}

For the linear Gaussian model of \eqref{eq:obs} this abstract condition reduces
to a transparent algebraic one. Stacking \eqref{eq:truthmean} gives the mean
vector $\mathbb{E}[\boldsymbol{m}] = A\,(\theta, s)^\top$ with design matrix:
\begin{equation}
  A \;=\; [\,\mathbf{1} \;\mid\; \boldsymbol{\alpha}\,] \in \mathbb{R}^{N \times 2},
  \qquad \mathbf{1} = (1,\dots,1)^\top .
  \label{eq:design}
\end{equation}

\begin{lemma}[Algebraic form of diversity]\label{lem:rank}
For the model \eqref{eq:obs}, SOC Sensor Diversity holds if and only if
$\operatorname{rank}(A) = 2$, the sensitivities $\alpha_i$ are not
all equal.
\end{lemma}
\begin{proof}
The variances are common knowledge and do not depend on $(\theta,s)$, so $L_{\theta,s}$
is determined by the mean vector $A(\theta,s)^\top$. Diversity is therefore injectivity
of $(\theta,s)\mapsto A(\theta,s)^\top$, which holds if and only if the columns
$\mathbf{1}$ and $\boldsymbol{\alpha}$ are linearly independent, that is
$\operatorname{rank}(A) = 2$, equivalently the $\alpha_i$ are not all equal.
\qed
\end{proof}

Diversity therefore states that the manipulation direction $\boldsymbol{\alpha}$ is
transverse to the all-ones direction $\mathbf{1}$, along which a change is
indistinguishable from a change in $\theta$. This geometric reading drives both
directions of the main result, and Proposition \ref{prop:invariance} formalises it as
location invariance. When the $\alpha_i$ are close but not equal, $A$ is
ill-conditioned and the variance of $\hat{s}$ grows towards degeneracy, which a
deployer counters by adding a sensor model with a distinct sensitivity profile. For a
richer feasible-action set the same argument applies to the expanded design
$[\,\mathbf{1}\mid\boldsymbol{\alpha}^{(1)}\mid\cdots\mid\boldsymbol{\alpha}^{(k)}\,]$.

%% file: 4_method.tex
\section{Equilibrium and Source Deterrence}
\label{sec:method}

\subsection{Mechanism Construction}

The mechanism has two stages: the first scores reporters, the second aggregates their
reports and settles the source. The contribution is the incentive characterisation,
not the estimator, which is ordinary \textbf{\underline{W}}eighted \textbf{\underline{L}}east \textbf{\underline{S}}quares (\textbf{WLS}).
CRPS peer scoring elicits truthful reports at the honest-source equilibrium, while the
diversity condition of Section \ref{sec:model} decides whether source manipulation is
identifiable and deterrable by report-only rules. Theorem \ref{thm:main} separates the
two claims.

\noindent \textbf{Reporter stage.}
Each reporter receives the peer payment defined after \eqref{eq:peercrps}. Because
CRPS is strictly proper and the mechanism scores each reporter's predictive for the
peer observation, a reporter minimises its peer CRPS only by reporting that
predictive truthfully (Proposition \ref{prop:sensors}).

\noindent \textbf{Source stage.}
Given the reports, the mechanism solves the WLS problem
induced by \eqref{eq:truthmean} and \eqref{eq:design}.
With weights $W = \operatorname{diag}(1/\sigma_1^2,\dots,1/\sigma_N^2)$ and
report means $\boldsymbol{m} = (m_1,\dots,m_N)^\top$, the estimator is:
\begin{equation}
  \big(\hat{\theta},\, \hat{s}\big)^\top
    \;=\; (A^\top W A)^{-1} A^\top W \boldsymbol{m},
  \label{eq:wls}
\end{equation}
which is defined exactly when $\operatorname{rank}(A) = 2$.
The mechanism issues credits proportional to $\hat\theta$ and charges the source
a penalty $P(\hat{s})$, where $P$ is non-decreasing and non-constant with
$P(z) = 0$ for $z \le 0$. The source action is one-sided, $s \ge 0$, since only
upward inflation raises the credited aggregate.
The source's payoff from action $s$, net of credits and penalty:
\begin{equation}
  U_S(s) \;=\; p\,\mathbb{E}[\hat\theta \mid s] \;-\; \mathbb{E}\big[P(\hat{s}) \mid s\big],
  \label{eq:srcutil}
\end{equation}
where $p > 0$ is the credit price. Crediting on $\hat\theta$ is deliberate. Under
diversity $\hat\theta$ is unbiased (Theorem \ref{thm:main}), so the credit term does
not vary with $s$ and manipulation yields no credit gain. The penalty then makes honest
sourcing strictly best rather than merely offsetting a gain. Section
\ref{sec:security} sizes the stake for the alternative deployment that credits the
more accurate but manipulable precision-weighted mean.
We take manipulation to be costless to the source, the adversarial worst case in which
the mechanism must deter cheating through scoring alone. An intervention that costs
the source energy or labour only makes deterrence easier, so the analysis needs no
estimate of that cost.

We adopt the continuous peer prediction setting \cite{richardson2024peer}, where the
variance inflation in \eqref{eq:peercrps} makes truthful reporting proper without a
separate posterior-consistency hypothesis. At the candidate equilibrium $s=0$, so the
sensitivities enter only through $\alpha_j s$ and diversity plays no role at the
reporter stage. Its role is at the source stage, where it makes $s$ identifiable.

\subsection{Main Theorem}

The necessary direction rests on a single structural property of CRPS.

\begin{lemma}[CRPS translation equivariance]\label{lem:equivariance}
For every $c \in \mathbb{R}$, every distribution $F$, and every $y$,
\begin{equation}
  \mathrm{CRPS}(F^{c}, y + c) \;=\; \mathrm{CRPS}(F, y),
\end{equation}
where $F^{c}(z) = F(z-c)$ is $F$ shifted by $c$.
Consequently, the peer score \eqref{eq:peercrps} depends on the report means only
through the difference $m_i - m_j$, and a common shift $m_k \mapsto m_k + c$
applied to all reporters leaves every peer score unchanged.
\end{lemma}

\begin{proof}
Substituting $u = z-c$ in \eqref{eq:crps} gives the identity, and \eqref{eq:peercrps}
depends on the means only through $d = (m_i-m_j)/\kappa$.
\qed
\end{proof}

\begin{proposition}[Strict reporter-stage equilibrium]\label{prop:sensors}
Under the standing assumptions, with the source playing $s = 0$ and all peers
reporting truthfully, reporting $R_i = (Y_i, \sigma_i^2)$ is the unique best
response for every reporter. Truthful reporting is therefore a strict Bayes--Nash
equilibrium of the reporter stage.
\end{proposition}

\begin{proof}
Strict propriety \cite{gneiting2007strictly} uniquely minimises each peer term at
reporter $i$'s predictive $\mathcal{N}(Y_i,\sigma_i^2+\sigma_j^2)$, forcing $m_i=Y_i$
and $v_i=\sigma_i^2$. The peer-specific $\sigma_j^2$ cancels, so one report minimises
every term at once.
\qed
\end{proof}

\smallskip\noindent\textbf{Scope of the guarantee.}
The result below is a strict Bayes--Nash statement on the equilibrium path. Three
limits follow.
First, a source deviation is evaluated holding the reporters' equilibrium strategy
fixed, under which the estimator identifies the deviation through $\hat s$. We claim
no sequential or perfect-Bayesian refinement. If reporters placed positive
probability on $s>0$ and re-optimised, their peer-predictive means would generally
differ from the raw $Y_i$ under heterogeneous sensitivities. That could attenuate the
pattern that identifies $s$.
Second, truthful reporting is strict against unilateral deviations, not unique. A
coalition that controls every reporter and adds a common shift $b\mathbf{1}$ to the
reported means raises the credit by $b$ while leaving $\hat s$ and, by
Lemma \ref{lem:equivariance}, every peer score unchanged. That is the collusion form
of the degeneracy the theorem rules out on the source side, and it is why a live
deployment needs sealed reports and Sybil-resistant identity. Every partial coalition
moves its reports away from its peers and pays a strictly positive peer-CRPS penalty.
Third, the impossibility direction restricts the class of settlement rules, as a rule
free to punish high reports as such would deter manipulation without paying for carbon.

\begin{definition}[Admissible settlement rule]\label{def:admissible}
A settlement rule is \emph{report-only} if the credit it pays and the penalty it
charges are functions of the report profile $R$ alone, and \emph{value-monotone} if
its honest-source net payoff $V(\theta)$, evaluated under truthful reporting at true
value $\theta$ with $s=0$, is non-decreasing in $\theta$. A rule with both properties
is \textbf{admissible}.
\end{definition}

\begin{lemma}[Stochastic monotonicity]\label{lem:monotone}
If $\hat s\mid s\sim\mathcal{N}(s,v_s)$, then $\hat s$ is stochastically increasing in
$s$. Hence for any non-decreasing $P$, $\mathbb{E}[P(\hat s)\mid s]$ is non-decreasing
in $s$, and strictly increasing whenever $P$ is non-constant.
\end{lemma}

\begin{theorem}[Source deterrence if and only if sensor diversity]\label{thm:main}
In the flat-prior Gaussian reporting model, CRPS peer scoring makes truthful
reporting a strict best response at the honest-source equilibrium. Fix truthful
reports by non-colluding reporters. There is then a penalty $P$ for which $s=0$ is a
strict source best response if and only if SOC Sensor Diversity
(Definition \ref{def:diversity}) holds. Necessity holds for every admissible
settlement rule (Definition \ref{def:admissible}) and uses no property of the scoring
rule. Under a proper common Gaussian prior the same conclusion follows after the
de-shrinking transformation of Section \ref{sec:invariance}.
\end{theorem}

\begin{proof}[sketch]
\emph{Sufficiency.} Under diversity $\operatorname{rank}(A)=2$ (Lemma
\ref{lem:rank}), so the WLS estimator \eqref{eq:wls} is unbiased and $\hat s$
identifies any source deviation taken against the fixed equilibrium reporting
strategy. The credit term in \eqref{eq:srcutil} is then constant in $s$, and a
penalty non-decreasing in $\hat s$ makes $s=0$ the strict best response by Lemma
\ref{lem:monotone}. \emph{Necessity.} If the sensitivities are all equal to
$\bar\alpha$, source action $s$ at true value $\theta_0$ and an honest source at the
higher true value $\theta_1=\theta_0+\bar\alpha s$ induce the identical report law,
$L_{\theta_0,s}=L_{\theta_1,0}$. A report-only rule pays the same credit and the same
penalty in both scenarios, and by translation equivariance
(Lemma \ref{lem:equivariance}) every peer score is unchanged as well. Value
monotonicity then gives
$U_S(s)|_{\theta_0}=V(\theta_0+\bar\alpha s)\ge V(\theta_0)=U_S(0)|_{\theta_0}$ for
every $s>0$, so no penalty makes $s=0$ strict. Neither step invokes a property of the
score, so the impossibility covers every scoring rule. The full argument is in
Appendix \ref{app:proofs}.
\qed
\end{proof}

The necessity argument has no finite analogue: the indistinguishability
$L_{\theta_0,s}=L_{\theta_1,0}$ holds in the continuous model directly, so the
characterisation follows by a different route than the discrete enumeration of
\cite{milionis2025incentive} rather than as a corollary of it. This is what strict
propriety does not supply. Propriety fixes what a reporter should say, whereas the
question here is whether the source action is present in the report law at all.

\smallskip\noindent\textbf{Extensions.}
The characterisation generalises to $k$ manipulation modes, where strict truthfulness
needs $\operatorname{rank}(A)=k+1$, and to a known noise covariance $\Sigma$, under
which the diagonal estimator stays exactly unbiased and only
$\operatorname{Var}(\hat s)=c_s^\top\Sigma\,c_s$ changes. A sensitivity drift biases
the credit only through its non-pro\-por\-tional part. Sufficiency uses only strict
propriety and mean elicitation, so the logarithmic and energy scores qualify, while
necessity uses no property of the score at all. The project webpage carries the
details.

\subsection{Invariance and Aggregation}
\label{sec:invariance}

Detecting manipulation is an invariant inference problem, which is why the prior on
$\theta$ plays no role in it. The location group acts on report means by
$g_c\,\boldsymbol m=\boldsymbol m+c\,\mathbf{1}$, and shifting $\theta\mapsto\theta+c$
shifts every truthful mean by $g_c$ through \eqref{eq:truthmean}, so $\theta$ is the
group parameter and the manipulation $s$ is the invariant part.

\begin{proposition}[Prior-free manipulation detection]\label{prop:invariance}
Let $\Pi=\mathbf{I}-\tfrac1N\mathbf{1}\mathbf{1}^\top$ and let $T=\Pi\boldsymbol m$ be
the contrast vector of the report means, the maximal location invariant. Under
truthful reporting $T\sim\mathcal{N}(s\,\Pi\boldsymbol\alpha,\,\Pi\Sigma\Pi)$, a law
free of $\theta$, and the row $c_s^\top$ of \eqref{eq:wls} satisfies
$c_s^\top\mathbf{1}=0$ and $c_s^\top\boldsymbol\alpha=1$, so $\hat s=c_s^\top T$. The
law of $\hat s$, and with it detection and deterrence, is therefore invariant to
$\theta$ and to any prior placed on it. Manipulation is identifiable from $T$ if
and only if $\Pi\boldsymbol\alpha\neq 0$, the condition of Lemma \ref{lem:rank}.
\end{proposition}

Only $\hat\theta$ is equivariant rather than invariant: it tracks $\mathbf{1}$ and
shrinks towards the prior mean under a proper prior. That shrinkage never reaches
$\hat s$, which is what makes deterrence indifferent to the source knowing $\theta$.
Under a proper common prior $\theta\sim\mathcal{N}(\mu_0,\tau^2)$ strict propriety
elicits $m_i=(1-\beta_i)\mu_0+\beta_i Y_i$ with
$\beta_i=\sigma_i^{-2}/(\tau^{-2}+\sigma_i^{-2})$, and de-shrinking to
$\hat Y_i=(m_i-(1-\beta_i)\mu_0)/\beta_i$ before \eqref{eq:wls} restores
$\mathbb{E}[\hat Y_i]=\theta+\alpha_i s$ with variance $\sigma_i^2$. The de-shrink is
a fixed affine map, so identification, the diversity condition, and the necessity
argument hold for every finite $\tau^2$.

The estimator \eqref{eq:wls} reuses the calibrated forecasts that the incentive layer
scores. Its honest-case limit is a known optimum. Dropping the source column from
\eqref{eq:wls} gives the precision-weighted mean with $w_i=1/\sigma_i^2$. That is the
minimum-variance unbiased linear aggregate by the Gauss--Markov theorem, and the
minimax-optimal pooling for the quadratic loss underlying CRPS
\cite{neyman2023proper}. The full estimator does not attain that minimum, since
identifying $s$ needs the second column of $A$. That is the price of robustness, and
Section \ref{sec:experiments} measures it.

\smallskip\noindent\textbf{On-chain settlement.}
The mechanism is a deterministic, closed-form function of the reports, so it suits
direct smart-contract execution. Reporters commit their reports, and the contract
evaluates $(\hat\theta,\hat s)$ and the peer payments, credits $\hat\theta$, and
slashes the source collateral when $\hat s$ exceeds the threshold. The published
$(\sigma_i,\alpha_i)$ also let it reject an ensemble whose precision-weighted
sensitivity variance falls below a floor, which turns diversity into a procurement
requirement. We stop at the rule itself, and claim no contribution on contract
implementation, gas or fixed-point error, report authentication, or the governance
that recalibrates $\alpha$ and $\sigma$.

%% file: 5_experiments.tex
\section{Experiments}
\label{sec:experiments}

\subsection{Experimental Setup}

The reporters are instantiated by eight machine learning models trained to predict
SOC from real visible-to-near-infrared (vis-NIR) reflectance spectra. They stand in
for the calibrated pipelines that reporters would run. The experiments therefore test
the statistical layer, source identifiability and the invariance of the credited value,
and not the strategic behaviour that Theorem \ref{thm:main} governs.
We use the LUCAS subset of the Open Soil Spectral Library, which links each
spectrum to a laboratory organic carbon (OC) measurement \cite{panagos2022european},
and keep the $37{,}280$ mineral soils with OC at most $100$ g/kg, the carbon-credit
deployment domain. This excludes extreme-carbon organic soils, a separate regime with
much larger prediction variance. Each spectrum is resampled to a $10$ nm grid over
$400$--$2450$ nm ($206$ bands). The eight models are a ridge linear baseline, a random
forest, a one-dimensional convolutional network, a spectral vision transformer,
SSL-SoilNet \cite{kakhani2024ssl}, a CNN-LSTM with efficient channel attention
\cite{wang2025soil}, a shifted-window (Swin) transformer \cite{jin2023innovative},
and a stacked autoencoder network \cite{vincent2010stacked}.

Two parameters of the model in Section \ref{sec:model} are measured, not assumed.
The noise level $\sigma_i$ is each model's five-fold cross-validation root mean
squared error.
The sensitivity $\alpha_i$ is each model's SOC-prediction response to a
soil-moisture perturbation of the spectrum, normalised so the largest is one.
Table \ref{tab:benchmark} reports both, with the prediction accuracy ($R^2$). The
measured $\alpha_i$ span $0.004$ to $1.000$, so the SOC Sensor Diversity Condition
holds empirically rather than by assumption.

\begin{table}[t]
\centering
\caption{\textbf{Sensor diversity holds empirically.} SOC benchmark on mineral LUCAS
vis-NIR ($37{,}280$ samples, five-fold cross-validation), with each model's measured
noise $\sigma_i$ and sensitivity $\alpha_i$, which span more than two orders of
magnitude. Best $R^2$ and lowest $\sigma_i$ in \best{bold}, second in \snd{underline}.}
\label{tab:benchmark}
\scriptsize
\setlength{\tabcolsep}{30pt}\renewcommand{\arraystretch}{1.25}
\begin{tabular}{@{}lccc@{}}
\toprule
Model & $R^2$ & $\sigma_i$ (g/kg) & $\alpha_i$ \\
\midrule
CNN                                        & \best{$0.505$} & \best{$13.22$} & $1.000$ \\
Random forest                              & \best{$0.505$} & \snd{$13.23$} & $0.418$ \\
Vision Transformer                         & \snd{$0.491$} & $13.41$ & $0.538$ \\
Swin transformer \cite{jin2023innovative}        & $0.471$ & $13.66$ & $0.705$ \\
Stacked autoencoder \cite{vincent2010stacked}  & $0.443$ & $14.02$ & $0.716$ \\
SSL-SoilNet \cite{kakhani2024ssl}          & $0.396$ & $14.60$ & $0.202$ \\
Ridge (linear)                             & $0.387$ & $14.72$ & $0.665$ \\
CNN-LSTM-ECA \cite{wang2025soil}        & $0.084$ & $17.98$ & $0.004$ \\
\bottomrule
\end{tabular}
\end{table}

The moisture perturbation models wetting as added water absorption, darkening the
spectrum and deepening the O-H bands at $1450$ and $1940$ nm. It is built from those
known absorption bands rather than from measured wet-and-dry pairs of the same soils,
which would be a stronger test and which no open dataset currently supplies. Seven of
the eight models predict lower OC under wetting, so credit inflation comes from
presenting drier soil. The sign does not affect the mechanism, which depends only on
the spread of $\alpha_i$. A source able to pick its intervention is treated in
Section \ref{sec:security}.

The mechanism experiments bootstrap ground truth $\theta$ from the LUCAS 2018 mineral
OC values ($17{,}231$ samples, median $19.9$ g/kg), and source manipulation is a
fraction of $\theta$. We report the systematic bias
$\mathbb{E}[(\hat\theta-\theta)/\theta]$, the over-crediting that manipulation
induces, and the mean absolute relative error
$\mathrm{MARE}=\mathbb{E}[|\hat\theta-\theta|/\theta]$, which also includes random
error. Each point averages $1{,}000$ draws over five seeds, ten for the significance
test.

\smallskip\noindent\textbf{Robustness to the weakest sensor model.}
CNN-LSTM-ECA has near-zero sensitivity and anchors the low end of the spread, so we
check how much rests on it. Dropping it leaves sensitivities spanning $0.202$ to
$1.000$, so the design stays full rank. The manipulation estimate only becomes
noisier, its standard deviation rising from $18.1$ to $22.3$ g/kg and detection power
at full manipulation falling from $29\%$ to $23\%$.

\subsection{Manipulation Resistance}

At the honest-source equilibrium the peer CRPS of each of the eight reporters,
computed analytically from \eqref{eq:peercrps}, is uniquely minimised at zero deviation
from the truthful report, which confirms Proposition \ref{prop:sensors} on the
measured, heterogeneous sensitivities.

\smallskip\noindent\textbf{Systematic bias.}
On the parametric model the WLS estimate holds a bias of $-1.3\%$ across the whole
manipulation range, while the precision-weighted and mean aggregators over-credit by
up to $+57\%$ and output agreement reaches $+31\%$ (Table \ref{tab:baselines}). That is the
quantity that matters for a carbon market, since systematic over-crediting is fraud
whereas random error averages out across parcels. The MARE is flat at $56.8\%$ under
diverse sensitivities but grows from $24.4\%$ to $101.8\%$ under uniform ones, and its
level is the noise floor of vis-NIR SOC prediction rather than a property of the
manipulation.

\smallskip\noindent\textbf{Trained models on perturbed spectra.}
We now drop the parametric $\theta+\alpha_i s$ model and run the eight trained models
directly on real LUCAS spectra perturbed by the wetting operator (train $29{,}824$,
test $7{,}456$). Two facts emerge. First, the models are biased predictors: on clean
spectra they over-predict OC by $7\%$ to $27\%$ (one under-predicts by $5\%$), so any
aggregate inherits a calibration bias. Second, the manipulation invariance survives on
these real nonlinear responses (Figure \ref{fig:realmanip}(a)): the precision-weighted,
mean, and median aggregators collapse from about $+14\%$ to between $-47\%$ and
$-55\%$, while the WLS estimate stays within a $9$-point band ($+17\%$ to $+26\%$),
about seven times less sensitive. The residual $\approx 20\%$ level is the models' own
calibration bias, not the manipulation. The mechanism therefore removes the
manipulation bias but inherits the calibration bias. The guarantee is relative
(manipulation invariance) rather than absolute (accuracy).
\begin{figure}[t]
  \centering
  \begin{minipage}[b]{0.47\linewidth}
    \centering
    \includegraphics[width=\linewidth]{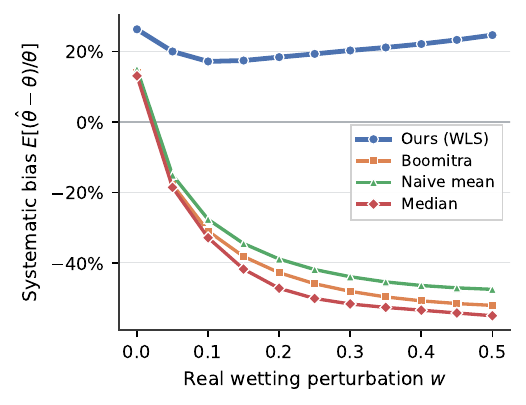}\\[2.5pt]
    {\footnotesize (a)}
  \end{minipage}\hfill
  \begin{minipage}[b]{0.47\linewidth}
    \centering
    \includegraphics[width=\linewidth]{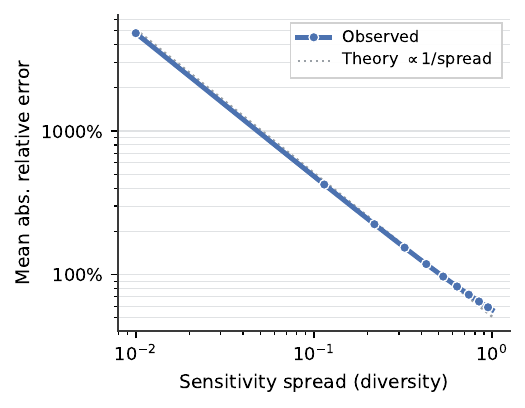}\\[2.5pt]
    {\footnotesize (b)}
  \end{minipage}
  \caption{\textbf{Invariance holds on real spectra; error diverges as diversity
  vanishes.} \textbf{(a)} Systematic bias of the eight trained models on real perturbed
  spectra; the aggregation rules collapse as wetting grows while WLS stays in a narrow
  band whose nonzero level is the models' own calibration bias. \textbf{(b)} Estimation
  error versus sensitivity spread at twenty percent manipulation, the empirical form of
  $v_s\to\infty$.}
  \label{fig:realmanip}\label{fig:diversity}
\end{figure}

\smallskip\noindent\textbf{Sensor diversity.}
Figure \ref{fig:diversity}(b) interpolates the sensitivities from uniform to the
measured diverse set under twenty percent manipulation. As the spread approaches
zero the error diverges, from $56.8\%$ at the diverse setting to over $4800\%$
near the uniform limit, the empirical boundary of Theorem \ref{thm:main}.

\subsection{Aggregation-Rule Comparison}

Table \ref{tab:baselines} compares the WLS mechanism against seven aggregation rules.
Each is the core aggregation rule of a deployed system (precision-weighted mean for
Boomitra, outlier removal and peer comparison for WeatherXM, output agreement for
OFP), abstracted from its full pipeline. The rows are not evaluations of those
pipelines. The comparison isolates the aggregation rule because that is what governs
manipulation resistance. Spectral moisture compensation and direct moisture sensing
sit upstream of the oracle and are complementary rather than competing
(Section \ref{sec:related}).
At full manipulation six of the seven rules over-credit by more than half and the
seventh by $31\%$, while WLS does not, and the absolute errors converge near the noise
floor. A paired one-sided $t$-test and Wilcoxon signed-rank test over ten seeds find
WLS significantly more accurate than six of the seven, with $t$-test $p$ from
$3{\times}10^{-6}$ to $5{\times}10^{-3}$, the Wilcoxon $p$ at its ten-seed floor
$9.8{\times}10^{-4}$ for all but two rows, and every mean absolute relative error
inside a $\pm2$-point $95\%$ confidence interval.
At zero manipulation the simple aggregators are more accurate, since separating the
source costs efficiency when there is nothing to separate. We report that trade-off
rather than hide it. Table \ref{tab:baselines} traces the frontier. MARE$_0$ is the
honest-source error, the variance cost, and bias$_1$ is the manipulation bias. WLS
sits at high MARE$_0$ and near-zero bias. The aggregators sit at low MARE$_0$ and large
bias.

\smallskip\noindent\textbf{Estimator versus incentives.}
The bias gap in Table \ref{tab:baselines} is a property of the estimator, not of the
payments. Any estimator that fits $(\theta,s)$ jointly inherits the same unbiasedness
once diversity holds: Huber and $\ell_1$ regressions on the same design hold a credit
bias of $-1.5\%$ against the $-1.3\%$ of WLS, so the invariance follows from modelling
$s$ and not from the loss or the payments. The peer scoring and the penalty do a
separate job. They make truthful reporting and $s=0$ a strict equilibrium
(Theorem \ref{thm:main}), so a rational source gains nothing from manipulating at all,
and the detected $\hat s$ supports the slashing analysed below.

\begin{table}[t]
\centering
\caption{\textbf{Shift-naive aggregation over-credits; the joint estimate does not.}
Mean absolute relative error at zero (MARE$_0$) and full (MARE$_1$) manipulation,
systematic bias at full manipulation (bias$_1$), and the paired one-sided $t$-test $p$
against our method over ten seeds. Best in \best{bold}, second in \snd{underline}.}
\label{tab:baselines}
\setlength{\tabcolsep}{9pt}\renewcommand{\arraystretch}{1.2}
\begin{tabular}{@{}lcccc@{}}
\toprule
Method & MARE$_{0}$ & MARE$_{1}$ & bias$_{1}$ & $t$-test $p$ \\
\midrule
Output agreement (OFP)         & $49\%$ & $60\%$ & \snd{$+31\%$} & $4.4{\times}10^{-5}$ \\
Precision-weighted (Boomitra)  & \best{$24\%$} & $61\%$ & $+57\%$ & $9.3{\times}10^{-5}$ \\
Outlier removal (WeatherXM)    & $28\%$ & $61\%$ & $+56\%$ & $1.1{\times}10^{-5}$ \\
Peer comparison (WeatherXM)    & $31\%$ & $62\%$ & $+56\%$ & $2.8{\times}10^{-6}$ \\
Sample median                  & $28\%$ & $61\%$ & $+55\%$ & $3.2{\times}10^{-5}$ \\
CRPS-reweighted mean           & \snd{$25\%$} & \snd{$58\%$} & $+54\%$ & $4.5{\times}10^{-3}$ \\
Unweighted mean                & \snd{$25\%$} & \snd{$58\%$} & $+53\%$ & $0.04$ \\
\midrule
\rowcolor{ourstint}
WLS \ours    & $57\%$ & \best{$57\%$} & \best{$-1.3\%$} & n/a \\
\bottomrule
\end{tabular}
\end{table}

\smallskip\noindent\textbf{Robustness.}
Several departures from the idealised model leave the invariance intact, with the full
sweeps on the project webpage. Scaling the ensemble from three to twenty reporters
lowers the estimation error and raises detection power at fifty percent manipulation
from $10\%$ to $17\%$. Student-$t$, lognormal, and uniform noise of the same scale
leave the MARE flat between $50\%$ and $57\%$, as does a fifty percent error in the
assumed $\sigma_i$, since the theory needs only mean-zero noise. Correlated residuals,
which the eight models incur by reading a common spectrum, sharpen detection rather
than weakening it, since $\hat s$ is a contrast that cancels the common component:
detection power at fifty percent manipulation rises from $14\%$ at $\rho=0$ to $47\%$
at $\rho=0.9$. The independent-noise figures are therefore the conservative case.

\subsection{Security Analysis}\label{sec:security}

The credit rule is a design choice, and the stake exists only for one side of it.
Crediting on the debiased estimate $\hat\theta$ removes the manipulation gain
entirely (Theorem \ref{thm:main}), so honest sourcing is strict with no stake, at
the cost of accuracy: $\hat\theta$ inherits the $57\%$ error floor of
Table \ref{tab:baselines} against $24\%$ for the precision-weighted mean. A deployer
who values honest-case accuracy may instead credit on the precision-weighted mean of
Section \ref{sec:invariance} and restore resistance with a stake $M$ slashed when
$\hat s>\delta$.

\smallskip\noindent\textbf{Detector variance and stake.}
The manipulation estimate $\hat s=c_s^\top\boldsymbol m$, with $c_s$ the second row of
$(A^\top W A)^{-1}A^\top W$, satisfies $\hat s\mid s\sim\mathcal N(s,v_s)$ with
$v_s=(\sum_i w_i(\alpha_i-\bar\alpha_w)^2)^{-1}$, $w_i=\sigma_i^{-2}$ and
$\bar\alpha_w=\sum_i w_i\alpha_i/\sum_i w_i$, so detection sharpens with the spread of
the $\alpha_i$ and $v_s\to\infty$ as that spread vanishes. A one-sided test at
false-alarm rate $a$ rejects honest sourcing when $\hat s>z_{1-a}\sqrt{v_s}$, giving
power $P_{\text{detect}}(s)=\Phi(s/\sqrt{v_s}-z_{1-a})$. Manipulation is then deterred
when $(P_{\text{detect}}(s)-a)M\ge G(s)$ for the precision-weighted credit gain
$G(s)=p\,\bar\alpha_w s$, hence a minimum stake
$M^\star(s)=G(s)/(P_{\text{detect}}(s)-a)$. The deterrent is the detector's power above
its size, not its raw power. As the spread vanishes $P_{\text{detect}}\to a$ and
$M^\star\to\infty$: no finite stake deters without taxing honest sources equally. That
is the economic form of the necessity direction.
With $\sqrt{v_s}=18.1$ g/kg at the median value $\theta=19.9$ g/kg and $a=5\%$,
detection power is $7.7\%$, $13.7\%$, and $29.4\%$ at manipulations of $20\%$, $50\%$,
and $100\%$ of $\theta$, and $M^\star/G$ falls from $36.8\times$ through $11.5\times$
to $4.1\times$. At \$15 per tonne of carbon dioxide on one hectare the median parcel
earns an honest source \$4{,}269 per round, a twenty percent manipulation gains
$G=\$483$, and the minimum deterring stake is \$17{,}785. The stake is an economic
attack-cost condition and a statistical rather than a cryptographic one: it raises the
cost of manipulation without making it infeasible.

\smallskip\noindent\textbf{Scope and limits.}
Four limits bound the guarantee, and none is repaired by a larger stake. The detector
is statistically thin. The standard deviation of $\hat s$ is about $18$ g/kg against a
median value of $20$, so detection power is modest and the stake layer needs the
worst-case gain $G$ estimated accurately. The no-bias property holds along
$\boldsymbol\alpha$, calibrated from a synthetic wetting operator on a single survey.
A source free to choose its intervention could pick one whose response is near-uniform,
close to $\mathrm{span}(\mathbf{1})$. That drives the ensemble towards the
rank-deficiency boundary that only diversity against every feasible direction rules
out. Whether the
measured sensitivities are stable across seasons and soil textures needs longitudinal
data and paired wet-and-dry field measurements that we do not have. Finally, a
coalition of all eight reporters adding a common shift pays no peer penalty and drives
$\hat s$ to zero, whereas every partial coalition pays a strictly positive peer-CRPS
penalty, so a live oracle needs a settlement layer that seals reports and preserves a
diversity floor under quorum selection.

%% file: 6_conclusion.tex
\section{Conclusion}
\label{sec:conclusion}

We extended source-manipulation-resistant peer prediction to continuous Gaussian
reports. Strict propriety of CRPS settles only what a reporter should say. The result
here is about the source. By translation equivariance, an admissible report-only
settlement rule identifies and deters source manipulation if and only if the sensor
sensitivities differ, a condition that holds on eight real SOC models. The source-stage
estimator is ordinary weighted least squares, and the characterisation around it is
what the paper contributes. The guarantee is limited. It holds on the
equilibrium path, assumes non-colluding reporters, and is calibrated along one
manipulation direction measured from a synthetic wetting operator. It delivers
manipulation invariance rather than absolute accuracy. Sensor diversity is therefore a
structural requirement for honest sourcing. Paired wet-and-dry validation, diversity
against an adversary that chooses its intervention, and a settlement layer that seals
reports are left to future work.

%% file: 7_appendix.tex
\newpage
\section{Deferred Proofs}
\label{app:proofs}

This appendix gives the full proofs of the results stated in the main text:
Theorem \ref{thm:main}, Lemma \ref{lem:equivariance}, Proposition
\ref{prop:sensors}, Lemma \ref{lem:monotone}, the honest-case aggregation
statement of Section \ref{sec:invariance}, the score-free impossibility
(Proposition \ref{prop:scorefree}), and the prior-free detection property
(Proposition \ref{prop:invariance}).

\begin{proof}[of Lemma \ref{lem:equivariance}]
Substituting $u = z - c$ in \eqref{eq:crps},
\begin{equation*}
\begin{aligned}
  \mathrm{CRPS}(F^{c}, y+c)
    &= \int \big(F(z-c) - \mathbf{1}\{z \ge y+c\}\big)^2 dz \\
    &= \int \big(F(u) - \mathbf{1}\{u \ge y\}\big)^2 du \\
    &= \mathrm{CRPS}(F, y).
\end{aligned}
\end{equation*}
In \eqref{eq:peercrps} the score is a function of $d = (m_i - m_j)/\kappa$ and of
$\sigma_i,\sigma_j$ alone, so it is invariant under $m_i \mapsto m_i + c$,
$m_j \mapsto m_j + c$.
\qed
\end{proof}

\begin{proof}[of Proposition \ref{prop:sensors}]
Fix truthful reporting by all peers $j \ne i$ and the equilibrium belief $s = 0$.
Sensor $i$ chooses $(m_i, v_i)$ to minimise
$\sum_{j \ne i} \mathbb{E}\big[\mathrm{CRPS}\big(\mathcal{N}(m_i,\, v_i + \sigma_j^2),\, Y_j\big) \mid Y_i\big]$.
Because CRPS is strictly proper \cite{gneiting2007strictly}, each term is uniquely
minimised when the scored distribution equals sensor $i$'s posterior predictive for
$Y_j$, which is $\mathcal{N}(Y_i,\, \sigma_i^2 + \sigma_j^2)$. This forces
$m_i = Y_i$ and $v_i + \sigma_j^2 = \sigma_i^2 + \sigma_j^2$, so $v_i = \sigma_i^2$.
The peer-specific $\sigma_j^2$ cancels, so the single report $(Y_i, \sigma_i^2)$
minimises every term at once and is the unique minimiser of the sum. Uniqueness
makes the equilibrium strict.
\qed
\end{proof}

\begin{proof}[of Lemma \ref{lem:monotone}]
For $\hat s\mid s\sim\mathcal N(s,v_s)$ the law is a location shift in $s$, so
$\Pr(\hat s>x\mid s)=\Phi\big((s-x)/\sqrt{v_s}\big)$ is strictly increasing in $s$ for
every $x$. Hence $\hat s$ is strictly stochastically increasing in $s$, and by
first-order stochastic dominance $\mathbb{E}[P(\hat s)\mid s]$ is non-decreasing in $s$
for every non-decreasing $P$, strictly so when $P$ is non-constant. The argument uses
no smoothness of $P$, so it covers the threshold slashing rule.
\qed
\end{proof}

\begin{proof}[of the honest-case aggregation statement, Section \ref{sec:invariance}]
With known $\sigma_i$ and unbiased Gaussian means, inverse-variance weighting is the
minimum-variance unbiased linear combination by the Gauss--Markov theorem, and the
minimax-optimal aggregate of \cite{neyman2023proper}. Setting the design matrix to
the single column $\mathbf{1}$ in \eqref{eq:wls} yields this weighting.
\qed
\end{proof}

\begin{proof}[of Theorem \ref{thm:main}]
\emph{Sufficiency} ($\Leftarrow$).
Assume diversity.
By Lemma \ref{lem:rank}, $\operatorname{rank}(A) = 2$, so the WLS estimator
\eqref{eq:wls} is well defined.
By Proposition \ref{prop:sensors} truthful reporting is a strict best response at
the equilibrium action $s = 0$, and a sensor's report equals its observation for
any realised $s$. Under the equilibrium reporting strategy,
source deviations change realised observations to $Y_i=\theta+\alpha_i s+\varepsilon_i$,
and sensors continue to report those realised observations.
Under truthful reporting $\mathbb{E}[\boldsymbol m] = A(\theta,s)^\top$, so the WLS
estimator is unbiased, $\mathbb{E}[\hat\theta \mid s] = \theta$ and
$\mathbb{E}[\hat s \mid s] = s$, and $\hat s$ identifies any source deviation under that fixed equilibrium reporting strategy.
It remains to give the source a strict incentive to play $s = 0$.
The credit term in \eqref{eq:srcutil} is constant in $s$ because
$\mathbb{E}[\hat\theta\mid s]=\theta$.
By Lemma \ref{lem:monotone}, $\mathbb{E}[P(\hat s)\mid s]$ is strictly increasing in
$s$, so on the one-sided action set $s\ge0$ it is minimised at $s=0$.
Hence $U_S(0) > U_S(s)$ for all $s > 0$, and $s=0$ is the strict best response.
Truthful reporting together with $s=0$ is a strict Bayes--Nash equilibrium.

\emph{Necessity} ($\Rightarrow$).
We argue the contrapositive.
Suppose diversity fails.
By Lemma \ref{lem:rank}, the $\alpha_i$ are all equal to a common value
$\bar\alpha$.
Thus $\boldsymbol\alpha = \bar\alpha\,\mathbf{1}$, and source action $s$ shifts
every report mean by the same amount $\bar\alpha s$.
Fix any true value $\theta_0$.
Under truthful reporting and source action $s$, the report means have distribution
$m_i \sim \mathcal{N}(\theta_0 + \bar\alpha s,\,\sigma_i^2)$.
This is exactly the profile produced by an honest source ($s'=0$) at the higher
true value $\theta_1 = \theta_0 + \bar\alpha s$.
The two scenarios induce the identical report law, $L_{\theta_0,s} = L_{\theta_1,0}$.
Two consequences follow.
First, by Lemma \ref{lem:equivariance} the shift is common to all sensors, so
every peer score is unchanged. The sensor payments carry no information about $s$.
Second, any mechanism is a function of the report profile alone.
Its credit and penalty terms therefore take the same values in the two scenarios.
Write $V(t) := U_S(0)\big|_{\theta=t}$ for the honest-source payoff at true value
$t$. For the payment rules considered here, $V$ is strictly increasing without
further assumptions: the credit pays $p$ times a location-equivariant aggregate of
the report means (the estimator \eqref{eq:wls} where defined, with the
precision-weighted mean as its honest-case limit), and the penalty reads the
reports only through a contrast, whose law is free of $\theta$
(Proposition \ref{prop:invariance}), so $V(t) = p\,t - \mathrm{const}$.
The source payoff satisfies
$U_S(s)\big|_{\theta_0} = V(\theta_0 + \bar\alpha s) > V(\theta_0) = U_S(0)\big|_{\theta_0}$
for all $s > 0$.
Manipulation is therefore strictly profitable, so $s = 0$ cannot be a strict best
response, and no penalty $P$ yields a strict equilibrium. Proposition
\ref{prop:scorefree} extends the impossibility to every report-based mechanism
whose honest-source net payoff is non-decreasing in the true value.
\qed
\end{proof}

\begin{proposition}[Score-free impossibility for admissible rules]\label{prop:scorefree}
Suppose $\boldsymbol\alpha\parallel\mathbf{1}$. For every scoring rule and every
admissible settlement rule in the sense of Definition \ref{def:admissible}, that is
every report-only rule whose honest-source net payoff $V$ is non-decreasing in the
true value, manipulation is weakly profitable, so $s=0$ is not a strict best
response. Without value monotonicity the statement fails: a rule that penalises high
reported values as such deters manipulation, but it no longer pays for carbon.
\end{proposition}

\begin{proof}
When $\boldsymbol\alpha\parallel\mathbf{1}$ the report laws satisfy
$L_{\theta_0,s}=L_{\theta_0+\bar\alpha s,0}$, shown in the proof of
Theorem \ref{thm:main}. A report-only mechanism assigns the two scenarios
identical credits and penalties, so the source payoff is
$V(\theta_0+\bar\alpha s)\ge V(\theta_0)$ by monotonicity. Sufficiency, by
contrast, needs only strict propriety and mean elicitation, so the CRPS, the
logarithmic score, and the multivariate energy score all qualify, the last
extending the characterisation to vector-valued reports.
\qed
\end{proof}

\begin{proof}[of Proposition \ref{prop:invariance}]
$T$ is the maximal invariant of the location group. Its mean is
$\Pi\,\mathbb{E}[\boldsymbol m]=\Pi(\theta\mathbf{1}+s\boldsymbol\alpha)
=s\,\Pi\boldsymbol\alpha$ since $\Pi\mathbf{1}=0$, and its covariance does not
involve $(\theta,s)$. The normal equations of \eqref{eq:wls} give
$c_s^\top A=(0,1)$, that is $c_s^\top\mathbf{1}=0$ and
$c_s^\top\boldsymbol\alpha=1$. From $c_s\perp\mathbf{1}$ we get $\Pi c_s=c_s$ and
$\hat s=c_s^\top\boldsymbol m=(\Pi c_s)^\top\boldsymbol m=c_s^\top T$. The mean
$s\,\Pi\boldsymbol\alpha$ separates the values of $s$ exactly when
$\boldsymbol\alpha\notin\operatorname{span}(\mathbf{1})$, which is
Lemma \ref{lem:rank}.
\qed
\end{proof}

\section{Stake Sizing}
\label{app:stake}

This appendix derives the collateral required for the high-efficiency
precision-weighted-credit deployment discussed in Section \ref{sec:security}, and makes explicit its
dependence on the sensitivity spread, the sensor noise, and the detector error
rates, the connection summarised in Section \ref{sec:experiments}.

\smallskip\noindent\textbf{Detector variance.}
The manipulation estimate is $\hat s = c_s^\top\boldsymbol m$, where $c_s$ is the
second row of $(A^\top W A)^{-1}A^\top W$. Under truthful reporting
$\hat s \mid s \sim \mathcal N(s, v_s)$ with
\begin{equation}
\begin{gathered}
  v_s \;=\; \big[(A^\top W A)^{-1}\big]_{22}
       \;=\; \frac{1}{\sum_i w_i\,(\alpha_i-\bar\alpha_w)^2},\\
  w_i=\tfrac{1}{\sigma_i^2},\qquad
  \bar\alpha_w=\frac{\sum_i w_i\alpha_i}{\sum_i w_i}.
\end{gathered}
\label{eq:vs}
\end{equation}
The denominator is the precision-weighted variance of the sensitivities, so
detection sharpens with the spread of the $\alpha_i$ and with lower noise, and
$v_s\to\infty$ as the spread vanishes (the uniform-sensitivity limit).

\smallskip\noindent\textbf{Detection power.}
A one-sided test rejects honest sourcing when $\hat s>\delta$. Fixing the
false-alarm (Type I) rate at $a$ sets $\delta=z_{1-a}\sqrt{v_s}$, and the detection
power (one minus the Type II rate) at manipulation $s$ is
\begin{equation}
  P_{\text{detect}}(s) \;=\; \Pr(\hat s>\delta\mid s)
    \;=\; \Phi\!\Big(\frac{s}{\sqrt{v_s}}-z_{1-a}\Big).
  \label{eq:pdetect}
\end{equation}

\smallskip\noindent\textbf{Deterring stakes.}
Section \ref{sec:security} gives the deterrence condition $(P_{\text{detect}}(s)-a)\,M\ge G(s)$ for the
precision-weighted credit gain $G(s)=p\bar\alpha_w s$. Substituting \eqref{eq:pdetect} gives the minimum
stake in closed form,
\begin{equation}
M^\star(s)=\frac{G(s)}{P_{\text{detect}}(s)-a}
        =\frac{p\bar\alpha_w s}{\Phi\!\big(s/\sqrt{v_s}-z_{1-a}\big)-a}.
\label{eq:stakeapp}
\end{equation}
The deterrent is the detector's power above its size, $P_{\text{detect}}(s)-a$, not its raw power. As the
sensitivity spread vanishes, $v_s\to\infty$ and $P_{\text{detect}}(s)\to a$, so $M^\star\to\infty$: no
finite stake deters without taxing honest sources equally, the economic form of the necessity direction.

The detection power and the minimum stake at the median value $\theta=19.9$ g/kg,
with $\sqrt{v_s}=18.1$ g/kg and a $5\%$ false-alarm rate, follow from
\eqref{eq:pdetect}--\eqref{eq:stakeapp} and are reported in Table \ref{tab:stakeapp}.

\begin{table}[ht]
\centering
\caption{\textbf{The minimum deterring stake falls from $36.8\times$ to $4.1\times$ the
gain as manipulation grows.} Detection power and minimum stake at the median SOC value
($\theta=19.9$ g/kg, $5\%$ false-alarm rate), with $\sqrt{v_s}=18.1$ g/kg.}
\label{tab:stakeapp}
\setlength{\tabcolsep}{28pt}\renewcommand{\arraystretch}{1.25}
\begin{tabular}{@{}lccc@{}}
\toprule
Manipulation $s/\theta$ & $20\%$ & $50\%$ & $100\%$ \\
\midrule
Detection power $P_{\text{detect}}$ & $7.7\%$ & $13.7\%$ & $29.4\%$ \\
Minimum stake $M^\star/G$ & $36.8\times$ & $11.5\times$ & $4.1\times$ \\
\bottomrule
\end{tabular}
\end{table}

\begin{figure}[ht]
  \centering
  \begin{minipage}[b]{0.49\linewidth}
    \centering
    \includegraphics[width=\linewidth]{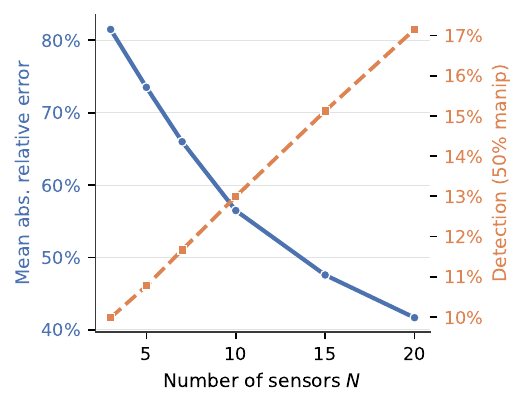}\\[2pt]
    {\footnotesize (a)}
  \end{minipage}\hfill
  \begin{minipage}[b]{0.49\linewidth}
    \centering
    \includegraphics[width=\linewidth]{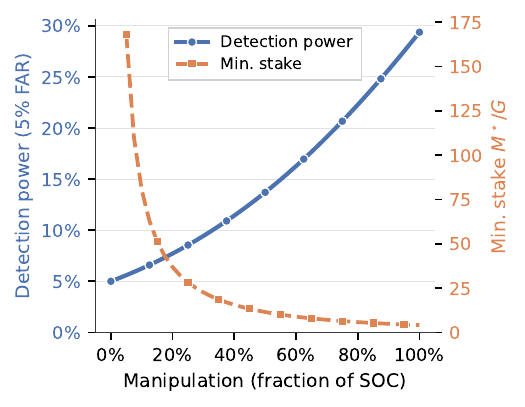}\\[2pt]
    {\footnotesize (b)}
  \end{minipage}
  \caption{\textbf{Scaling and stake.} \textbf{(a)} Estimation error (left axis) and
  detection power at fifty percent manipulation (right axis) versus the number of
  sensors; the honest and manipulated errors coincide. \textbf{(b)} Detection power
  (left axis) and minimum deterring stake $M^\star/G$ (right axis) versus
  manipulation for diverse sensors; uniform sensors give no power above the
  false-alarm rate, so no finite stake selectively deters them.}
  \label{fig:scaling}\label{fig:stake}
\end{figure}

\noindent Figure \ref{fig:stake} reports the detection power and stake versus
manipulation in panel (b), and the gain in accuracy and detection from a larger
ensemble in panel (a).

\section{Reporter-Stage Robustness under Unknown Covariance}
\label{app:robust}
Section \ref{sec:method} uses mean reports that are truthful at the
flat-prior honest-source equilibrium. The following lemma isolates the limited
robustness property needed below: misspecifying the scored variance does not move
the CRPS-optimal mean away from the conditional mean of the target distribution.

\begin{lemma}[Mean pinning]\label{lem:meanpin}
Let $Y\sim\mathcal N(\mu,\tau^2)$ and fix any $\omega>0$. The expected score
$h(m)=\mathbb E_Y[\mathrm{CRPS}(\mathcal N(m,\omega^2),Y)]$ is strictly convex in
$m$ with the unique minimiser $m=\mu$, independent of $\omega$.
\end{lemma}

\begin{proof}
The Gaussian CRPS is
$\mathrm{CRPS}(\mathcal N(m,\omega^2),y)=\omega\,g((y-m)/\omega)$ with
$g(z)=z(2\Phi(z)-1)+2\phi(z)-\pi^{-1/2}$, so $g'(z)=2\Phi(z)-1$. Hence
$h'(m)=-\mathbb E_Y[g'((Y-m)/\omega)]=1-2\,\mathbb E_Y[\Phi((Y-m)/\omega)]$. With an
independent $Z\sim\mathcal N(0,1)$,
$\mathbb E_Y[\Phi((Y-m)/\omega)]=\Pr(m\le Y-\omega Z)
=\Phi\big((\mu-m)/\sqrt{\tau^2+\omega^2}\big)$, so
$h'(m)=2\Phi\big((m-\mu)/\sqrt{\tau^2+\omega^2}\big)-1$ and
$h''(m)=2\phi\big((m-\mu)/\sqrt{\tau^2+\omega^2}\big)/\sqrt{\tau^2+\omega^2}>0$.
Thus $h$ is strictly convex and $h'(m)=0$ exactly at $m=\mu$.
\qed
\end{proof}

\noindent At the reporter stage the target variable in Lemma \ref{lem:meanpin} is the peer observation
$Y_j$, conditional on sensor $i$'s information. The lemma therefore says that
a misspecified scored variance does not move the optimal reported mean away
from the conditional mean
$\mathbb E[Y_j\mid Y_i]$.
In the flat-prior honest-source model this conditional mean is $Y_i$. Hence,
in that baseline case, variance-inflation errors leave the elicited mean pinned
at $m_i=Y_i$. This is the only sense in which the mean report is robust to
variance misspecification: Lemma \ref{lem:meanpin} does not assert that $Y_i$ is optimal under
arbitrary beliefs about source manipulation or under an arbitrary prior. In those
cases the optimal mean is the appropriate posterior-predictive mean.

With known covariance $\Sigma$, sensor $i$'s posterior predictive distribution
for peer $j$ in the diffuse-prior honest-source model is
\begin{equation}
    Y_j\mid Y_i
    \sim
    N\!\left(
        Y_i,\,
        \operatorname{Var}(Y_j-Y_i)
    \right)
    =
    N\!\left(
        Y_i,\,
        \sigma_i^2+\sigma_j^2-2\Sigma_{ij}
    \right).
\end{equation}
Thus the covariance-corrected version of the peer score replaces the independent
inflation in \eqref{eq:peercrps} by the pre-calibrated variance
$\sigma_i^2+\sigma_j^2-2\Sigma_{ij}.$
Equivalently, if the report variance is written as $v_i$, the scored variance is
$v_i+\sigma_j^2-2\Sigma_{ij}$ and at truthful reporting $v_i=\sigma_i^2$.
This variance is positive unless the two sensors coincide almost surely.

If the mechanism instead uses a conservative or misspecified positive variance
inflation, Lemma \ref{lem:meanpin} implies that the elicited mean remains the conditional mean
of $Y_j\mid Y_i$. Therefore, in the flat-prior honest-source model, the mean
input to the WLS estimator remains $Y_i$ even when the variance inflation is
misspecified. Only the variance report is distorted. Section \ref{sec:experiments} confirms this
numerically for the equicorrelated family, where the optimal mean stays within
the search-grid resolution across correlations up to $\rho=0.9$.

For general $\Sigma$, the diversity condition and the impossibility argument
are unchanged because they depend only on the mean design $A=[{\bf 1}\mid
\alpha]$. The efficient estimator is the generalised least squares estimator
\begin{equation}
    (A^\top\Sigma^{-1}A)^{-1}A^\top\Sigma^{-1}m,
\end{equation}
while the diagonal WLS estimator \eqref{eq:wls} remains unbiased for $(\theta,s)$ under
any true covariance because its rows $c_\theta^\top,c_s^\top$ satisfy
\begin{equation}
    c_\theta^\top A=(1,0),\qquad c_s^\top A=(0,1).
\end{equation}
The covariance affects precision through
\begin{equation}
    \operatorname{Var}(\hat s)=c_s^\top\Sigma c_s,
\end{equation}
and similarly for $\hat\theta$. Common-mode covariance can reduce
$\operatorname{Var}(\hat s)$ because $c_s^\top{\bf 1}=0$, whereas covariance
aligned with the contrast direction $c_s$ can increase it. Thus covariance affects
detection power and stake sizing, but not the algebraic identifiability condition
or the unbiasedness identities.

\section{Coalition Manipulation}
\label{app:collusion}

A coalition $C$ of reporters that adds a bias vector $\boldsymbol b$ (supported on
$C$) to its reported means shifts the estimate by
$(\Delta\hat\theta,\Delta\hat s)=(c_\theta^\top\boldsymbol b,\;c_s^\top\boldsymbol b)$,
linear in $\boldsymbol b$, where $c_\theta,c_s$ are the rows of
$(A^\top W A)^{-1}A^\top W$. Because $c_\theta^\top\mathbf 1=1$ and
$c_s^\top\mathbf 1=0$, a coalition that controls every sensor and adds a common bias
$b\mathbf 1$ inflates the credit by the full $b$ while leaving $\hat s$, and by
translation equivariance every peer score, unchanged: the all-network common shift
is the only manipulation that is both undetected and unpenalised, identical to the
uniform-sensitivity failure of Theorem \ref{thm:main}. Any partial coalition
($C\ne\{1,\dots,N\}$) moves its reports away from its peers and pays a strictly
positive peer-CRPS penalty, which falls with coalition size and reaches zero only at
the full network (Figure \ref{fig:collusion}(b)). There is thus no hard tolerance
threshold: truthful reporting degrades gracefully, and a single deviant peer
perturbs an honest sensor's optimal mean by only about $b/(N-1)$
(Section \ref{sec:method}). Sybil-resistant identity, which prevents one party from
controlling the whole ensemble, confines manipulation to the penalised partial
regime.

\begin{figure}[ht]
  \centering
  \begin{minipage}[b]{0.48\linewidth}
    \centering
    \includegraphics[width=\linewidth]{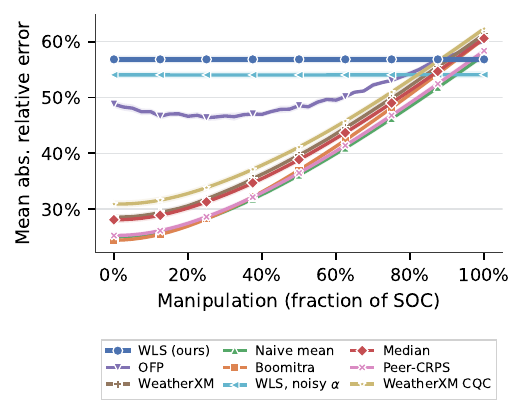}\\[2.5pt]
    {\footnotesize (a)}
  \end{minipage}\hfill
  \begin{minipage}[b]{0.51\linewidth}
    \centering
    \includegraphics[width=\linewidth]{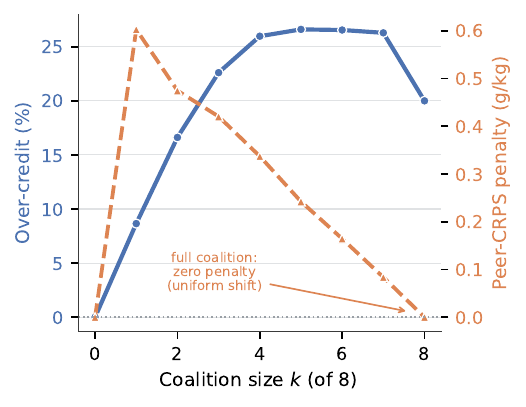}\\[2.5pt]
    {\footnotesize (b)}
  \end{minipage}
  \caption{\textbf{Baselines over-credit; collusion is penalised by the peer score.}
  \textbf{(a)} Absolute error versus manipulation for all methods; the curves converge
  near the noise floor. \textbf{(b)} A coalition of $k$ reporters adds a common bias to
  inflate the credit; every partial coalition pays a peer-CRPS penalty, and only the
  all-eight common shift pays zero, the collusion form of the necessity boundary.}
  \label{fig:baselines}\label{fig:collusion}
\end{figure}

\noindent Figure \ref{fig:baselines}(a) shows the absolute errors of all methods
converging near the noise floor at full manipulation.

\section{Regression Baselines}
\label{app:regression}

A regression that models the source shift as a nuisance parameter, fitting
$(\theta,s)$ on the design $[\mathbf 1\mid\boldsymbol\alpha]$, inherits the
manipulation invariance whatever loss it minimises. Table \ref{tab:regression} fits
$(\theta,s)$ by weighted least squares (the estimator of \eqref{eq:wls}), a Huber
loss, and an $\ell_1$ loss. Each holds a credit bias between $-1.3\%$ and $-1.5\%$ that is constant
across the manipulation range, while the shift-naive aggregators of
Table \ref{tab:baselines} reach $+57\%$. The invariance is therefore a property of
the estimator structure, not of the loss or of the payments. The incentive layer
does a different job: it does not lower the bias, which any shift-modelling
regression already removes, but makes truthful reporting and $s=0$ a strict
equilibrium (Theorem \ref{thm:main}) and supplies the detected $\hat s$ that the
on-chain stake slashes (Appendix \ref{app:stake}).

\begin{table}[ht]
\centering
\caption{Shift-modelling regressions are manipulation-invariant. Credit bias
$\mathbb E[(\hat\theta-\theta)/\theta]$ for estimators that jointly fit
$(\theta,s)$, constant across $0$ to $100\%$ manipulation, against the $+57\%$ of
the shift-naive baselines (Table \ref{tab:baselines}).}
\label{tab:regression}
\setlength{\tabcolsep}{10pt}\renewcommand{\arraystretch}{1.25}
\begin{tabular}{@{}lc@{}}
\toprule
Estimator (fits $\theta$ and $s$) & Credit bias \\
\midrule
\rowcolor{ourstint}
Weighted least squares \eqref{eq:wls} & $-1.3\%$ \\
Huber & $-1.5\%$ \\
$\ell_1$ & $-1.5\%$ \\
\bottomrule
\end{tabular}
\end{table}

\section{Calibration and Deployment}
\label{app:deploy}

\smallskip\noindent\textbf{Sensitivity drift.}
The most common calibration change, a uniform rescaling of all $\alpha_i$ from
sensor ageing, leaves $\hat\theta$ exactly
unbiased, because it does not change the column space of the design
(Section \ref{sec:method}). Only the non-proportional part of a drift biases the
credit, and that bias is bounded by
$|s|\,\|\tilde c_\theta\|\,\|\boldsymbol\Delta\|$ and stays below the baseline
manipulation bias even at a $50\%$ drift (Section \ref{sec:experiments}). The
sensitivities should be re-measured periodically on trusted reference soils rather
than source-controlled sites, so re-calibration opens no manipulation channel.

\smallskip\noindent\textbf{Verifying diversity on-chain.}
The contract holds the published $(\sigma_i,\alpha_i)$ and accepts an ensemble only
when $\operatorname{rank}(A)=2$, equivalently when the precision-weighted variance
$\sum_i w_i(\alpha_i-\bar\alpha_w)^2$ of \eqref{eq:vs} exceeds a floor. Sensor
heterogeneity is then a verifiable procurement requirement rather than an
assumption, and the same quantity sets the detection power and the stake.

\smallskip\noindent\textbf{Limits.}
The measured $\alpha_i$ are a single-survey snapshot. Their stability across seasons
and soil textures, and a validation against paired wet-and-dry field measurements,
need longitudinal data that we leave to deployment.

\section{Extensions}
\label{app:extensions}

\smallskip\noindent\textbf{Multiple manipulation modes.}
With $k$ modes $s\in\mathbb{R}^k_{\ge0}$ and
$Y_i=\theta+\sum_{j=1}^{k}\alpha_i^{(j)}s_j+\varepsilon_i$, the means stack into the
design $A=[\,\mathbf{1}\mid\boldsymbol\alpha^{(1)}\mid\cdots\mid
\boldsymbol\alpha^{(k)}\,]\in\mathbb{R}^{N\times(k+1)}$. The same argument shows
strict truthfulness is achievable if and only if $\operatorname{rank}(A)=k+1$: the
all-ones vector lies outside the span of the manipulation directions and those
directions are linearly independent. This requires $N\ge k+1$ sufficiently diverse
sensors, and Theorem \ref{thm:main} is the case $k=1$. As in the scalar case the
action is one-sided, and a penalty non-decreasing in each $\hat s_j$ makes $s=0$ the
strict best response by the componentwise form of Lemma \ref{lem:monotone}, in line
with the one-sided slashing event. If a sensitivity varies with
the true value or the manipulation level, the design $A$ becomes state dependent and
the characterisation applies to its local linearisation. The real-model results of
Section \ref{sec:experiments} show the invariance tolerates the mild spectral
nonlinearity of wetting.

\smallskip\noindent\textbf{Variance elicitation.}
The peer score \eqref{eq:peercrps} inflates the reported variance $v_i$ by the
common-knowledge peer noise $\sigma_j^2$, so strict propriety sets
$v_i+\sigma_j^2=\sigma_i^2+\sigma_j^2$ and the mechanism elicits calibrated
uncertainty along with the mean. The estimator \eqref{eq:wls} instead weights by the
pre-calibrated $\sigma_i^2$ of Table \ref{tab:benchmark}, not re-elicited on-chain.

\section{Parametric Manipulation Results}
\label{app:parametric}

Figure \ref{fig:bias} shows the manipulation resistance on the parametric
$\theta+\alpha_i s$ model summarised in Section \ref{sec:experiments}: the WLS bias
stays near zero across the manipulation range while the baseline aggregators
accumulate over-crediting, and the absolute error is flat under diverse
sensitivities but grows under uniform ones.

\begin{figure}[ht]
  \centering
  \begin{minipage}[b]{0.49\linewidth}
    \centering
    \includegraphics[width=\linewidth]{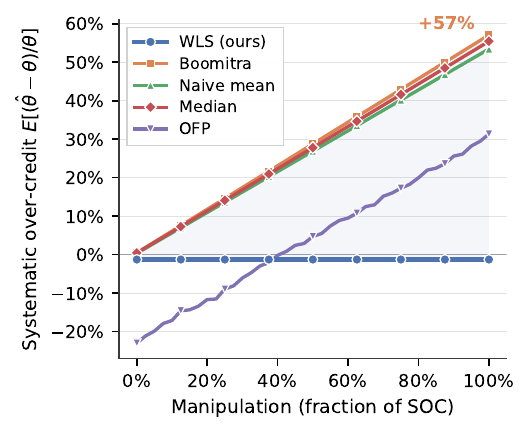}\\[2.5pt]
    {\footnotesize (a)}
  \end{minipage}\hfill
  \begin{minipage}[b]{0.49\linewidth}
    \centering
    \includegraphics[width=\linewidth]{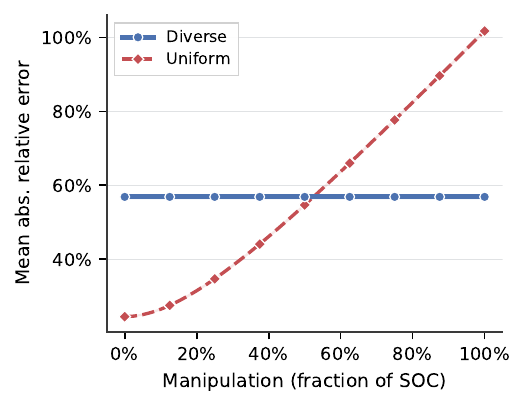}\\[2.5pt]
    {\footnotesize (b)}
  \end{minipage}
  \caption{\textbf{Manipulation resistance.} \textbf{(a)} Systematic bias versus
  manipulation fraction; baselines accumulate over-crediting while WLS stays near
  zero. \textbf{(b)} Absolute error versus manipulation fraction; flat under diverse
  sensitivities, growing under uniform ones.}
  \label{fig:bias}\label{fig:manipulation}
\end{figure}